\documentclass[11pt,dvipsnames]{article}
\usepackage[T1]{fontenc}
\usepackage{latexsym,epsf,graphics,graphpap,amsmath,amssymb,amsthm,mathtools}
\usepackage{graphicx}
\usepackage{xcolor}
\usepackage{hyperref}
\hypersetup{
colorlinks=true,
linkcolor=RoyalBlue,
citecolor=red}

\newtheorem{theorem}{Theorem}[section]

\newtheorem{definition}{Definition}[section]
\newtheorem{proposition}{Proposition}[section]
\newenvironment{defi}[1]{\begin{definition}{\bf#1}\rm}{\end{definition}}
\newtheorem{lemma}{Lemma}[section]

\newcommand{\R}{\mathbb R}

\newcommand{\s}{\ensuremath{\mathcal{S}}}

\newcommand{\ew}{\color{black}}

\newcommand{\rbt}{\ensuremath{\mathcal{T}_{0}^{2}}}
\newcommand{\be}{\begin{equation}}
\newcommand{\ee}{\end{equation}}
\newcommand{\bsplit}{\begin{split}}
\newcommand{\esplit}{\end{split}}

\newcommand{\bea}{\begin{eqnarray*}}
\newcommand{\eea}{\end{eqnarray*}}

\DeclareMathOperator\arctanh{arctanh}

\usepackage[normalem]{ulem}

\newcommand{\magnprime}[2]{\langle \sigma'_{#1} \rangle^{#2}}

\begin{document}

\title{{\bf Almost Gibbsian Measures on a {Cayley Tree}}}

 \author {Matteo D'Achille\ew\footnote{LAMA UPEC $\&$ CNRS, Universit\'e Paris-Est,  94010 Cr\'eteil, France,
 \newline
 email:  matteo.dachille@u-pec.fr}\ ,\; Arnaud Le Ny\ew \footnote{LAMA UPEC $\&$ CNRS, Universit\'e Paris-Est,  94010 Cr\'eteil, France,
 \newline
 email:  arnaud.le-ny@u-pec.fr.}
}

\maketitle

\begin{center}
{\bf Abstract:}
\end{center}

 We consider the ferromagnetic n.n.~Ising model on Cayley trees submitted to a modified majority rule transformation with overlapping cells already known to lead to non-Gibbsian measures. We describe the renormalized measures within the generalized Gibbs framework and prove that they are almost Gibbs at any temperature.

\bigskip
\footnotesize



\normalsize
\section{Introduction}

In mathematical statistical mechanics, {\em Gibbs measures} have been rigorously designed to represent equilibrium states and to model phase transitions, following the pioneering works of Dobrushin, Lanford and Ruelle who described them as an extension of Markov chains, both dynamically and spatially, in terms of the specification of their conditional probabilities w.r.t.\ the outside of finite sets \cite{Dob68,LanRue69}. This {\em DLR approach} has been fully put on rigorous grounds in the eighties by Georgii \cite{HOG}. In the mean time, some pathologies of the formal definition of Gibbs measures arose within numerical studies of critical phenomena \cite{GP,I}, identified afterwards to be a manifestation of possible non-Gibbsianness of the renormalized measures \cite{VEFS}. This non-Gibbsianness has been since then mainly coined by the exhibition of bad configurations, which are points of essential discontinuity of the renormalized measures. This led to a Dobrushin program of restoration of Gibbsianness, launched by Dobrushin himself in a talk in Renkum in 1995 \cite{DOB}. Within this program, two main restoration notions have been proposed, {\em weak Gibbsianness} and {\em almost Gibbsianness}, the latter being stronger than the former\footnote{See \cite{Maes2} or \cite{KLNR04}, where one can learn that the weakly Gibbsian representation is indeed too weak due to a possible failure of the variational principle in a weak but non almost Gibbsian example.}. On $d$-dimensional integer lattices $\mathbb{Z}^{d}$, most initial efforts of restoration have been adressed to the decimation transformation, leading to an almost Gibbsian description at any temperature in a series of papers \cite{VEFS, FLNR03, ALN2}, while positive results concerning preservation of Gibbsianness near the critical point and in uniqueness regions have been proposed in \cite{Bertini,HKen, Ken,Martinelli}.

\smallskip

As it is usually difficult to evaluate the measure of the set of discontinuity points (so-called {\em bad configurations}), the approach leading to almost Gibbsianness on $\mathbb{Z}^{d}$ heavily relies on a specification-dependent variational principle ({\em zero relative entropy}  characterization) which is known to fail on trees  \cite{Bu,foll}. Due to this difficulty, while considering trees as lattices, most efforts have been addressed to the detection of non-Gibbsian measures, by renormalization \cite{HK} or {\em via} stochastic evolution (van Enter {\em et al.}, see {\em e.g.}\ \cite{vEEIK2012}), with the notable exception of the failure of almost Gibbsianness for the random-cluster measures proved in \cite{Hagg2}. While decimation is known to preserve Gibbs property on trees, we focus here on a majority rule transformation already known from \cite{ALN1} to cause failure of Gibbsianness, and prove that {the renormalized measure} is almost Gibbsian at any temperature (Theorem \ref{thm3}). Our strategy is elementary in that it involves transfer matrices, Markov chains, dependent percolation {or} moment estimates, and profits of the recursivity inherent to treatments on trees.

\section{Gibbsian and non-Gibbsian Measures on Trees}
\subsection{Gibbs measures}

In this paper we consider Ising spins with single-spin state-space $E=\{-1,1\}$ equipped with the $\sigma$-algebra $\mathcal{E}=\mathcal{P}(E)$ and the {\em a priori} couting measure $\rho_{0}=\frac{1}{2}\delta_{-1}+\frac{1}{2}\delta_{+1}$, where $\delta_i$ is the Dirac measure on $i \in E$.  For a given lattice $S$, the configuration space is the product measurable space $\Omega=E^{S}$ equipped with the product $\sigma$-algebra $\mathcal{F}=\mathcal{E}^{\otimes S}$ and the product measure $\rho=\rho_0^{\otimes S}$. As usual in statistical mechanics, we consider macroscopic states to be probability measures on $\Omega$, whose set is denoted by $\mathcal{M}^{+}_{1}(\Omega)$. We also denote by $\mathcal{S}$ the set of all the finite subsets of $S$ and for $V \in \s$, we consider the finite-volume configuration space $\Omega_V=E^V$, equipped with the product measurable structure $(\mathcal{F}_{V},\rho_{V})$, and denote $\omega_{V}$ the canonical projection of $\omega \in \Omega$ on $\Omega_{V}$. Similarly, for any ${V},{V}' \subset {S}$ such that ${V} \cap {V}' \neq \emptyset$, for all $\omega, \sigma \in \Omega$, we denote by $\omega_{V}\sigma_{{V}'}$ the element of $\Omega_{{V} \cup {V}'}$ which agrees with $\omega$ in ${V}$ and with $\sigma$ in ${V}'$. For any $V \in \mathcal{S}$, $|{V}|$ denotes the cardinality of ${V}$.

\smallskip

In this work, the lattices we consider are Cayley trees $S=\mathcal{T}^{k}$, for integers $k>0$,  that is, $(k+1)$-regular infinite trees (see \cite{Roz13} for {further details}). We focus here on the case $k=2$.

\bigskip

{\bf Ising Potential:} we consider nearest-neighbor ($n.n.$) potentials $\Psi = (\Psi_{A})_{A \in \mathcal{S}}$ s.t.,

\be \label{eq.isingpot}
\forall \omega \in \Omega,\Psi_{A}(\omega) = \begin{cases}
-J(i,j)\; \omega_{i}\omega_{j} & {\rm if} \; A = \{i,j\}\\
-h(i)\; \omega_{i} & {\rm if} \; A=\{i\} \\
0 & {\rm otherwise}.
\end{cases}
\ee
Here, $J: \mathcal{T}^{k}\times \mathcal{T}^{k} \longrightarrow  \mathbb{R}$ is the coupling function, {assumed to be non-{negative} in this ferromagnetic set-up}, and $h:\mathcal{T}^{k} \longrightarrow  \mathbb{R}$ is often called an external magnetic field. The Ising potential {(\ref{eq.isingpot})} is a prototype of uniform absolutely convergent ``potentials'' used to define Gibbs measures in this mathematical framework \cite{HOG}:

\begin{defi}{[UAC potential]} A potential $\Phi=\big(\Phi_A\big)_{A \in \mathcal{S}}$ is a family of local functions $\Phi_A$ that are $\mathcal{F}_A$-measurable. It is said to be {\em uniformly absolutely convergent (UAC)} iff
$$
\sum_{{ A \ni i,} A \in \mathcal{S}} \sup_{\omega \in \Omega} |\Phi_{A}(\omega)| < +\infty, \quad \forall i \in {S}.
$$
\end{defi}

\smallskip
For such a UAC potential $\Phi$, and for all configurations $\sigma \in \Omega$, we introduce the finite-volume Hamiltonian with boundary condition $\omega \ \in \Omega$, defined by
\be \label{eq.finvhbc}
{\mathcal{H}^{\Phi}_{V}(\sigma \mid \omega)} \coloneqq
  \sum_{A \in \mathcal{S},A \cap {V} \neq  \emptyset } \Phi_{A} (\sigma_{V} \omega_{{V}^{c}}).
\ee

In the particular case of free boundary conditions at finite-volume $V$, one writes 

$$
\mathcal{H}^{\Phi,f}_{V}(\sigma) \coloneqq 
  \sum_{A \subset {V}} \Phi_{A} (\sigma_{V}).
$$

Associated to (\ref{eq.finvhbc}), there are, at temperature $\beta^{-1} > 0$, Boltzmann--Gibbs weights $e^{- \beta {\mathcal{H}^{\Phi}_{V}(\sigma \mid \omega)}}$ and corresponding partition functions at finite-volume $V$ and boundary condition $\omega$,
\be
Z^{\beta \Phi}_{V}({\omega)} = \int_{\Omega_{V}}e^{- \beta \mathcal{H}^{\Phi}_{V}(\sigma \mid \omega)} \rho_{V}(d\sigma_{V}),
\ee
where  $\rho_{V}\coloneqq \rho_0^{\otimes {V}}$ is the {\em a priori} product measure at finite-volume $V$.

\begin{defi}{[Gibbs specifications]} 
For an UAC potential $\Phi$, the set of probability kernels $\gamma^{\beta \Phi} =  (\gamma^{\beta \Phi}_{{V}})_{{V} \in \mathcal{S}}$, defined for all ${V} \in \mathcal{S}$ and $\sigma,\omega \in \Omega$, {as}
\be
\gamma_{V}^{\beta \Phi}(\sigma \mid \omega)\coloneqq \frac{1}{Z^{\beta \Phi}_{V}(\omega)}e^{- \beta \mathcal{H}^{\Phi}_{V}(\sigma \mid \omega)},
\ee
is called a {\it Gibbs specification} for potential $\Phi$ at inverse temperature $\beta$.

\end{defi}
In the case of free boundary conditions, we write uniformly in $\omega \in \Omega$, for all $\sigma \in \Omega$,
$$
\gamma_V^{\beta \Phi,f}(\sigma)=\gamma_{V}^{\beta \Phi,f}(\sigma \mid \omega)\coloneqq \frac{1}{Z^{\beta \Phi, f}_{V}}e^{- \beta \mathcal{H}^{\Phi,f}_{V}(\sigma)}.
$$

More generally, a specification $\gamma=\big( \gamma_V \big)_{V \in \mathcal{S}}$ is a family of probability kernels satisfying extra properties (properness and consistency, see \cite{Pres,HOG,VEFS}) so that they can represent a regular system of conditional probabilities of {\em some} probability measures $\mu$, in such a way that the {\em DLR equations}
\be \label{DLR}
\mu(\sigma \mid \mathcal{F}_{V^c}) (\cdot) = \gamma_{V}(\sigma \mid \cdot), \; \mu-{\rm a.s.},
\ee
are valid for any $V \in \s$ {and} $\sigma \in \Omega$.

\smallskip

This {\em DLR approach} to describe probability measures on infinite product probability spaces is crucial in our framework because it allows the definition of many different measures specified by the same specification. If the latter situation occurs, we say that there is a {\em phase transition}.

\begin{defi}{[Gibbs Measures]} 
A Gibbs measure $\mu$ is a probability measure on $\Omega$ for which the DLR equations (\ref{DLR}) are valid for a Gibbs specification $\gamma=\gamma^{\beta \Phi}$, at some inverse temperature $\beta >0$ and for a UAC potential $\Phi$. Alternatively, it is a probability measure $\mu$ which is invariant under the action of the kernels $\gamma_V^{\beta \Phi}$ for any finite-volume $V$,  
$$
\mu = \mu \gamma_V^{\beta \Phi}, \forall V \in \s.
$$

\end{defi}
The set of {\bf Gibbs measures} specified by the specification $\gamma^{\beta \Phi}$ is denoted by $\mathcal{G}(\gamma^{\beta \Phi})$, and the study of its  properties is a central aim in rigorous statistical mechanics, see \cite{HOG,VEFS,ALNENS}. It is a convex set ({more precisely}, a Choquet simplex) and one is mainly interested in the study of its extreme elements, called extremal measures (or sometimes {``states''}).

\subsection{Gibbs measures on Cayley trees}

Ising models on Cayley trees $\mathcal{T}^{k}$ have been first {rigorously} shown to exhibit phase transition by Preston \cite{Pres2} in 1974. The critical temperature has been explicitely identified in 1989 by Lyons by exploiting close links with branching processes due to the tree structures, leading to the by-now famous formula for the critical inverse temperature, {namely}
$$
\beta_{\rm c}=\arctanh{1/k}.
$$
The set of translation-invariant extremal Gibbs measures has been studied by Bleher {\em et al.} in the early nineties \cite{BL, BLG}, see also {\em e.g.}~\cite{spi,vEEIK2012}.

\begin{proposition}{[Gibbs measures of Ising models on Cayley trees, see~\cite{Pres2}]}

\begin{enumerate}
\item Let  $\mu^{+},\mu^{-}$ be the limiting measures\footnote{In the sense of a convergence along a net directed by inclusion, denoted $V \uparrow \s$, see \cite{HOG}.} of Ising specifications  with all $+$ and all $-$ boundary conditions,
$$
\mu^+ (\cdot) \coloneqq \lim_{V \uparrow  \s} \gamma_V^{\beta \Psi}(\cdot | +) \; {\rm and} \;  \mu^- (\cdot) \coloneqq \lim_{V \uparrow  \s} \gamma_V^{\beta \Psi}(\cdot | -) .
$$

Then $\mu^{+},\mu^{-} \in \mathcal{G}(\gamma^{\beta  \Psi})$ are extremal at all inverse temperatures $\beta >0$. 

%$\bullet$
\item  Let $\mu^{\#}$ be the limiting measure of Ising specifications with free boundary conditions, 
$$
\mu^{\#}(\cdot) \coloneqq \lim_{V \uparrow   \s} \gamma_V^{\beta \Psi,f}(\cdot) {\; .}
$$

Then $\mu^{\#}$ is extremal at high temperatures, and becomes non-extremal at low temperatures. The transition occurs at the {\it spin-glass} inverse temperature given by 
$$
\beta_{\rm SG}=\arctanh{1/\sqrt{k}}.
$$
\end{enumerate}
\end{proposition}

The particular question of the extremality and extremal decomposition of the free boundary condition case has been a long standing question studied since the seventies \cite{BL, BLG, GRS12, GMRS20, jof, ioffe2, EKPS,Miya}. A considerable number of extremal Gibbs measures has been first constructed  by Bleher
and Ganikhodjaev:
\begin{theorem}~\cite{BLG} \label{BLG}
For $\beta > \beta_c$, the number of extreme points of
$\mathcal{G}(\gamma^{\beta \Psi})$ is uncountable.
\end{theorem}

 These extreme points can be selected by uncountably many different
boundary conditions for which ``half'' of the Cayley tree is occupied
by the ``plus'' and the other half by the ``minus''. Note that Higuchi \cite{hig} constructed earlier other non-translation-invariant
extreme points. It has been an open
question for a long time to know if we have then described {\em all} the extreme
points of $\mathcal{G}(\gamma)$, and to derive the convex decomposition of any Gibbs measures w.r.t.\ these extremal ones. 
A decade ago, another uncountable family of non-translation invariant extremal Gibbs measures different from the Bleher--Ganikhodjaev and the Higuchi ones have been described by Akin {\em et al.}\ \cite{ART}, while the concept of weakly periodic Gibbs measures has been introduced by Rakhmatullaev and Rozikov \cite{RR08} to describe yet other non-translation invariant extremal states. These extremal and non-translation invariant Gibbs measures where described particularly by Gandolfi {\em et al.} as a manifold in \cite{GRS12} in a framework that eventually led recently to the convex decomposition of the so-called {\em free measure} $\mu^{\#}$ into extremal points  below  the spin-glass transition temperature~\cite{GMRS20}.

\subsection{A criterion for Gibbsianness: Quasilocality } An important feature in the theory of Gibbs measures is {\it quasilocality} which, heuristically, ensures that a measurement done at a precise point in the system is not impacted too much by far away perturbations. 
%{\bpu More precisely, we have the following (see~\cite{ALNENS})} 

\begin{defi}{[Quasilocal function]}
A function $f \colon \Omega \to \mathbb R$ is said to be \emph{quasilocal} if it is the uniform limit of a sequence of local\footnote{Recall that a function $f \colon \Omega \to \mathbb R$ is said to be \emph{local} if $\exists V \in \mathcal{S}$ s.t.\ $f$ is $\mathcal{F}_{V}$-measurable (that is, $f$ depends only on a finite number of spins). We denote by $\mathcal{F}_{{\rm loc}}$ the set of such functions.} functions $f_{n}$, {\em i.e.} 
\begin{displaymath}
\lim_{n \to \infty} \sup_{\omega \in \Omega} \mid f_{n}(\omega) - f(\omega) \mid = 0 .
\end{displaymath}
\end{defi}

A specification $\gamma$ is {\it quasilocal} if its kernels leave invariant the set of quasilocal functions $\mathcal{F}_{\rm qloc}$; this means in particular that for any local function $f$ and finite set $V$, the function $\gamma_V f$ defined for any boundary condition $\omega$ as
$$
\gamma_{V}f (\omega)= \int_\Omega f(\sigma) \gamma_V (d\sigma \mid \omega) 
$$
is quasilocal. 

In particular, for any finite-volume $V$, $\gamma_{V}f$ is a {\em continuous} function for any local function $f \in \mathcal{F}_{\rm loc}$. A measure is quasilocal if it is specified by a quasilocal specification. Remarkably, it has been shown by Kozlov~\cite{ko} and Sullivan~\cite{sull} that quasilocality plus a natural positivity requirement ({\it non-nullness}\footnote{A specification $\gamma$ is {\it non-null} if $\rho(A)>0 \Longrightarrow$ $\gamma_{V}(A \mid \omega) >0$  for all ${V} \in \s, A \in \mathcal{F}, \omega \in \Omega$. Non-nullness has been sometimes named ``finite-energy condition''.}) fully characterize Gibbs measures. Thus, if $\Phi$ is a UAC potential, any Gibbs specification $\gamma^{\beta \Phi}$ is quasilocal, so that {\bf any Gibbs measure is quasilocal}. In particular, continuity properties are preserved and one can learn in {\em e.g.}\ \cite{VEFS} that conditional expectations of local functions do not admit points of essential discontinuity (see also~\cite{Fer05,ALNENS} for details).

\smallskip

Being quasilocal for a Gibbs measure means in particular  that it is not possible to build a version of the conditional expectation of a local function which would be discontinuous as a function of the boundary condition. More precisely we have the following
\begin{defi}{[Point of essential discontinuity]}\label{def.ped}
%Let $\mathcal{M}^{+}_{1}(\Omega,\mathcal{F})$ denote the set of probability measures on $\Omega$.
%~\footnote{Such a set, which models the macroscopic states of the system, has been sometimes called the set of {\it random fields} on $\Omega$, see~\cite{ALNENS}).}.
A configuration $\omega^*$ is a {\em point of essential discontinuity} for a measure $\mu \in \mathcal{M}^{+}_{1}(\Omega)$ if $\exists {V} \in \mathcal{S}$, a function $f \in \mathcal{F}_{\rm loc}$, $\delta > 0$ and two neighborhoods of $\omega^*$, denoted $\mathcal{N}^{1}_{{V}}(\omega^*),\, \mathcal{N}^{2}_{{V}}(\omega^*)$, {and defined as
$$
\mathcal{N}^{i}_{{V}}(\omega^*) = \left\lbrace \sigma \in \Omega \; \text{s.t.} \; \sigma_{V}= \omega^{*}_{V}\right\rbrace, \quad i=1,2,
$$
}
such that
$$
\forall \xi \in \mathcal{N}^{1}_{{V}}(\omega^*), \; \forall \zeta \in \mathcal{N}^{2}_{{V}}(\omega^*), \quad \left | \mu\left[f  \mid \mathcal{F}_{{V}^{c}}\right] (\xi)-\mu\left[f  \mid \mathcal{F}_{{V}^{c}}\right](\zeta) \right| > \delta.
$$ 
\end{defi}

%\But the converse is not true (see \cite{A,BB} for counterexamples)}.

%\medskip 

For our purposes, the relevance of the Kozlov--Sullivan characterization, summarized by
 $$
 {\rm Gibbsian} \implies   {\rm quasilocal} \implies {\rm {essential \; continuity}}{,}
 $$
comes from the fact that, for non-null Gibbs measures {with finite state-space}, continuity is equivalent to quasilocality \cite{HOG}. Thus, a {Gibbs} measure for Ising spins does not admit points of essential discontinuity. As a consequence, this Kozlov--Sullivan characterization can be used as a proxy to prove non-Gibbsianness: it suffices to exhibit essential discontinuity (in the sense of Definition~\ref{def.ped}) of  conditional expectations as a function of the boundary conditions (for example, magnetization at a given site in our Ising context).
\smallskip

 In this paper we are concerned with the size of the set of such essential discontinuity points and show that these discontinuity points form a null-set at any temperature, rendering thus the transformed measure {\em almost Gibbsian}, in the sense {recalled} below.

\subsection{Almost Gibbsian and weakly Gibbsian measures}

For a given specification $\gamma$, we denote by $\Omega_\gamma$ its set of points of essential continuity ({\em good configurations}). Of course, for a Gibbs specification, $\Omega_\gamma=\Omega$. 

\begin{defi}{[Almost Gibbsian measure]} A probability measure $\mu \in \mathcal{M}_{1}(\Omega)$ is {\it almost Gibbsian} if there exists a specification $\gamma$ such that
$$
\mu \in \mathcal{G}({{\gamma}})\quad {\rm and} \quad \mu(\Omega_{\gamma})=1.
$$
\end{defi}

This is equivalent to the statement that there exist regular versions of finite-volume conditional probabilities of $\mu$ which are continuous as functions of the boundary conditions, possibly except on a zero-measure set. 
\begin{defi}{[Weakly Gibbsian measure]} A probability measure $\mu \in \mathcal{M}_{1}(\Omega)$ is {\it weakly Gibbsian} if $\mu \in \mathcal{G}({\gamma^{\beta \Phi}})$ with a potential $\Phi$ converging on a full set ($\mu(\Omega_\Phi)=1$), {\em i.e.} s.t.
$$
\forall V \in \mathcal{S},\; \sum_{A \cap {V} \neq \emptyset} \Phi_A(\omega) < \infty, \; {\rm for} \; \mu-{\rm a.e.}(\omega).
$$
\end{defi}

Note that an almost-sure version of the Kozlov-Sullivan representation (see \cite{Maes2}) allows to reconstruct an almost surely convergent potential consistent with an almost Gibbsian  measure, so that
almost Gibbs implies weakly Gibbs.

\section{{R}enormalization transformation and  results}
\subsection{Modified majority rule on a Cayley tree}

{Let us consider the binary Cayley tree at $k=2$ (our discussion is easily adaptable to higher $k\geq 3$). Upon deletion of an arbitrary edge $\langle i_0i_1 \rangle$ from the
Cayley tree $\mathcal{T}^2$, we get two (identical) rooted Cayley
trees, call them $\mathcal{T}_{0}^2$ and $\mathcal{T}_{1}^2$, for which the following holds:

\begin{theorem}{\cite{BLG}}\label{thm.split} \label{BlG}
A measure $\mu \in \mathcal{M}_{1}(\Omega)$ is an extreme Gibbs distribution on
$\mathcal{T}^2$ if and only if there exist 
extreme Gibbs distributions $\mu_0,\mu_1$ on $\mathcal{T}_{0}^2$ and $\mathcal{T}_{1}^2$
respectively, such that the following {splitting property} holds 
%{for any infinite-volume configuration $\sigma \in \Omega$},
\begin{displaymath}
\forall A \in \mathcal{F},\; \mu (A)=\frac{1}{Z} \int_A  \mu_0 ({\sigma}_{\mathcal{T}_{0}^2})\mu_1 ({\sigma}_{\mathcal{T}_{1}^2})e^{\beta J\sigma_{i_0}\sigma_{i_1}} \rho(d \sigma),
\end{displaymath}
where $Z$ is a normalizing constant depending on $\beta$. In such a case, $\mu_0$ and $\mu_1$ are uniquely determined by $\mu$.
%$\diamond$
\end{theorem}}

%{\sl {\bf Modified majority-rule on a Cayley tree}}:\\

It is thus equivalent to our purposes to work on {\em rooted} Cayley trees instead of the full Cayley tree. Now, let $\mu$ be any Gibbs measure for the Ising 
model on the rooted Cayley tree
$\mathcal{T}^{2}_{0}$. 
We choose the root as the origin and we
denote it $r$. Define
\begin{displaymath}
\Omega=\{-1,+1\}^{\mathcal{T}^{2}_{0}} \; \textrm{and} \;
\Omega'=\{-1,0,+1\}^{\mathcal{T}^{2}_{0}}.
\end{displaymath}
Let $R$ be any non-negative integer. We define the \emph{closed ball} of
radius $R$,
$
V_R=\{i \in \mathcal{T}^{2}_{0} \mid d(r,i) \leq R \}
$
and the \emph{sphere} of radius $R$ (or sometimes level or generation $R$),
$
W_R=\{i \in \mathcal{T}^{2}_{0} \mid d(r,i) =R \},
$
where $d$ is the canonical metric on $\mathcal{T}^{2}_{0}$.
%defined in Eq.~\eqref{eq.canmetric}.
Vertices of $\mathcal{T}^{2}_{0}$ {are represented} by sequences
of bits by means of the following recurrence:

\begin{itemize}
\item
Origin $r$ is represented by the void sequence and its neighbors by
$0$ and $1$.
\item
Let $R>0$ and let $i \in W_R$ be represented by $i^{*}$. The representations of the neighbors of $i$ in $W_{R+1}$, called 
 $k$ and $l$, are
$
k^*=i^*0 \; \textrm{and} \; l^*=i^*1.
$
\end{itemize}
In this way we obtain a representation of all the vertices of $\mathcal{T}^{2}_{0}$
(see Figure~\ref{fig.binreprcayley}). {To avoid notational overloading we shall use the same symbol $i$ for {both} the vertex
or its binary representation $i^*$}. Define the root cell $C_r=\{r,0,1\}$ and {let}, $\forall j \in \mathcal{T}^{2}_{0},\; j \neq
r$ a general cell $C_j=\{j,j0,j1\}$, where $j0$ and $j1$ are the two neighbors ({or children}) of $j$ from the ``following''
level (for example, $C_0=\{0,00,01\}$, see Figure~\ref{fig.TactionT02}).
Define as well $c_j=|C_j|$.
\begin{figure}[!hbtp]
\centering
\includegraphics[width=\textwidth]{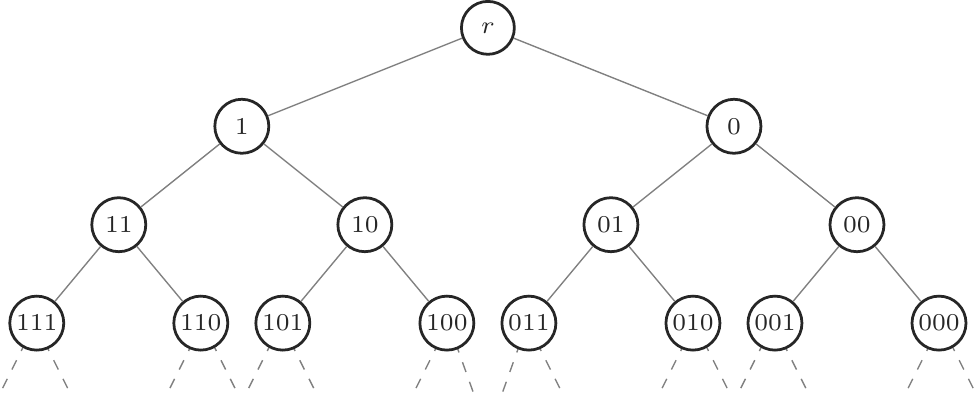} 
\caption{Binary representation of the vertices of $\mathcal{T}_{0}^{2}$, the rooted Cayley
tree of degree $2$. {Here vertices up to level/generation $R=3$ are represented}.}
\label{fig.binreprcayley}
\end{figure}

{In this paper we shall consider} $c_j=c=3$.

\smallskip

 Let us consider now the following, deterministic \emph{Majority rule transformation}:
\begin{eqnarray*}
T &\colon& \Omega \longrightarrow \Omega'\\
& & \omega \longmapsto \omega',
\end{eqnarray*}
where the {image spin $\omega'$ at site $j$} is defined by
\begin{displaymath}
\label{def.mmj}
\omega'_j=\left\{
\begin{array}{lllll}\; +1 \; & \textrm{iff}& \frac{1}{c}\sum_{i \in C_j}
\omega_i=+1\\
\\
\; 0 \; & \textrm{iff}&\frac{1}{c} \mid \sum_{i \in C_j}
\omega_i \mid < 1\\
\\
\; -1 \;& \textrm{iff}& \frac{1}{c}\sum_{i \in C_j}
\omega_i=-1.
\end{array} \right.
\end{displaymath}
%\begin{displaymath}
%\epsfbox{block.eps}
%\end{displaymath}

\begin{figure}[!htbp]
\centering
\includegraphics[width=\textwidth]{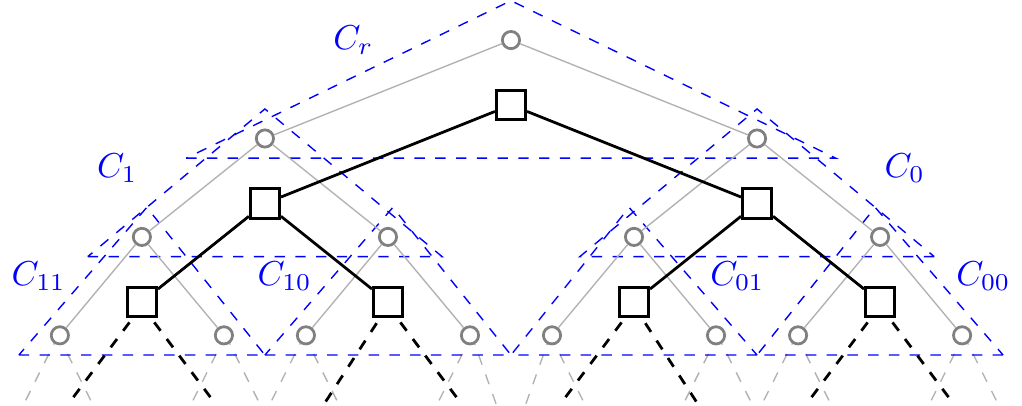} 
\caption{Pictorial representation of the transformation $T$ acting on $\mathcal{T}_{0}^{2}$. Starting vertices are represented by small circles and image vertices by squares.}
\label{fig.TactionT02}
\end{figure}

%{\bw {\em Action sur les mesures de Gibbs}} $\nu$ etc. Dire si on les prend tous ou si on se retreint à mu plus.
This transformation extends naturally on measures. We define  $\nu \coloneqq T \mu$ to be the {\em renormalized measure} defined as the image measure of our Gibbs measure $\mu$, so that
$$
\forall A' \in \mathcal{F}',\; \nu(A') = \mu(T^{-1} A')
$$
where $T^{-1} A' \in \mathcal{F}$ is the pre-image of the measurable set $A'$,
$$
T^{-1} A' = \big\{ \omega \in \Omega, T(w) \in A' \big\}
$$
%We shall use this transformation in the next section in order to
%describe examples where non-Gibbsianness arises.
%\subsection{Majority-rule on a Cayley tree}
\subsection{Previous result: Non-Gibbsianness at any temperature}

In \cite{ALN1}, one of us proved that the image measure of the modified majority rule $T$ on a rooted Cayley tree $\mathcal{T}_{0}^{2}$ of degree
$2$ is  not Gibbsian:
\begin{theorem}\cite{ALN1}
Let $\mu$ be any Gibbs measure for the Ising model on
$\mathcal{T}^{2}_{0}$. Then the renormalized {measure} $\nu=T\mu$ is non quasilocal at any $\beta >0$ and cannot be a Gibbs
measure.
%$\diamond$
\end{theorem}

The result of \cite{ALN1} is due to the following
\begin{lemma}
The ``null-configuration'' $\omega'^0$, defined by $\omega'^0_i= 0$  at any site $i \in \mathcal{T}_0^2$, is a bad configuration, {\em i.e.} a point of essential
discontinuity for the conditional probabilities of the image measure $\nu$ under the  majority rule $T$.
%$\diamond$
\end{lemma}
%This proves %, using cofinal sequences and the lemma 2.3.2,
%that {$\nu=T\mu$} is not
% quasilocal: it is not a Gibbs measure.
 
 \smallskip

Thus, we know that the ``null configuration'' is a bad configuration and that the image measure is not quasilocal. To incorporate the renormalized measure within the Dobrushin program, our purpose is to evaluate the measure of the set of such ``bad configurations''.

\section{Main result: almost Gibbsianness at any temperature}
{The main contribution of this paper is the following

\begin{theorem}\label{thm3}
The renormalized measures $\nu = T \mu$, obtained by the modified majority rule from any Gibbs measure {$\mu$} for the Ising model on $\mathcal{T}_0^2$,  are {\em almost Gibbsian at any temperature}.
\end{theorem}

%{\bf {Proof of Theorem \ref{thm3}}.}
%\label{sec.proof}
\begin{proof}
In the non-Gibbsian result of the previous section, the essential discontinuity of the bad  null  configuration $\omega'^0$ defined above is due to the massive dissemination of information possible through a neutral (filled of zero) cell, for which no specific discrimination between plus or minus is made. In turns, it appears that the crucial point to get essential discontinuity {\em via} transmission of far away information is that the configuration possesses enough {\em percolation of zeros}.

\smallskip

We proceed as follows: first, we recall some necessary percolation notions (paths and percolations of zeros) and describe the essential continuity of the magnetization in absence of such {a} percolation ({Lemma \ref{thm.esscontmagnbc}}). {Second,} we extend this continuity result to the case of a single infinite path of zeros ({Lemma \ref{thm.contmagn}}), and use its proof to get a sufficient condition for discontinuity coined in terms of the number of infinite such paths, {which should increase sufficiently fast with the volume} ({Lemma \ref{lem.suffcondec}}). {Finally}, we use percolation techniques (Lemma \ref{thm.expdb}) to upper bound the measure of bad configurations by zero (Lemma \ref{Omegac}), so that almost Gibbsianness follows. 

\end{proof}

\begin{defi}{[Path of zeros]}
{
{Let $\eta'$ be {any} image configuration.}
For an integer $R$, a {\em path of zeros from the origin to level $R$} is the sequence of sites $\pi=(i_k)_{k=0, \ldots , R}$ s.t.\ $i_0=r,\; \textrm{and} \; \forall k=1 \dots R, \; i_k \in W_k,\; i_k,i_{k-1} \; \textrm{are} \;  n.n. 
$ 
and $\forall k \leq R, \; \eta'(i_k)=0$.
}
%$\diamond$
\end{defi}
%{Note that the total number of paths of zeroes ending at level $R$ is $2^{R}$}.\\

We denote by $N_R(\eta')$ the number of paths of zeros connecting the origin to level $R$ in the configuration $\eta'$, and the number of infinite paths by $N(\eta')=\lim_{R \to \infty} N_{R}(\eta')$. {When} $N(\eta') \neq 0$, we say that there is {\em percolation} or {{\em clusters}} of zeros. 

\smallskip

{We first remark that in absence of percolation of zeros} in a configuration $\eta'$, {conditional magnetizations are (essentially) continuous, in the sense that there always exists a version of the conditional expectation of the spin at the origin which is continuous}.

\subsection{Continuity of the magnetization in absence of percolation}\label{sss.cn0}

Consider first the case of absence of percolation of zeros, that is, {let us take} a configuration $\eta'$ such that $N(\eta')=0$. Then there exists an integer
$R_0 \geq 0$ such that for all $R \geq R_0=R_0(\eta'), \, N_R(\eta')=0$.
Let us prove that any such {}$\eta'$ is a good
configuration for the image measure $\nu=T\mu$. In order to do so, consider the basis of
neighb{o}rhoods $\mathcal{N}^{R}(\eta')$ for {a} configuration $\eta'$, 

\begin{displaymath}
\mathcal{N}^{R}(\eta')=\{\omega' \in \Omega', \;
\omega'_{V_{R}}=\eta'_{V_{R}},\;\omega' \;
\textrm{arbitrary elsewhere}\}.
\end{displaymath}
We shall study the continuity, as a function of the boundary
condition, of the value of the image spin at the origin $r$ in the neighb{o}rhood
of $\eta$, namely the {conditional} magnetization
\be \label{condmagn1}
\langle \sigma'_r \rangle^{\eta',R}=\nu[\sigma'_r \mid
\sigma'_{\{r\}^c}=\omega'_{\{r\}^c}, \; \omega' \in
\mathcal{N}^{R}(\eta')].
\ee
%We shall prove the following result :
\begin{lemma}{\cite{ALN1}}\label{lem.levelindep} Let $\eta'$ be an image configuration s.t.\ $N(\eta')=0$. Then, for all $R \geq R_0(\eta')$, the conditional magnetization {(\ref{condmagn1})} 
reduces to 
$$
\langle \sigma'_r \rangle^{\eta',R}=\nu[\sigma'_r
\mid \sigma'_{V_{R_0 \setminus \{r\}}}=\eta'_{V_{R_0 \setminus \{r\}}}]
$$
and it is independent of $R \geq R_0$.
%$\diamond$
\end{lemma}
{The proof of Lemma \ref{lem.levelindep} has been reported in \cite{ALN1}. It consists of direct computations of elementary conditional expectations in this discrete framework.}
% and is recalled in Appendix~\ref{app:lviproof} for convenience.}
{It follows that} the magnetization in a neighb{o}rhood of such an $\eta'$ is a
function of $\eta'_{V_{R_0 \setminus \{r\}}}$, and {thus} essentially  continuous:
\begin{lemma}{\cite{ALN1}}
\label{thm.esscontmagnbc}
The magnetization is essentially continuous as a function of the
boundary condition {$\eta'$}, for all  $\eta' \in T(\Omega)$ without percolation of zeros.
\end{lemma}

Proceeding similarly for any local functions by integrating out w.r.t.~spins at the sites of the dependence set of $f$, one get this essential continuity for the expectation of any local function. 

\subsection{Continuity of the magnetization in case of unique infinite cluster}

{In this section we refine the analysis and estimate the effect that an infinite path of zeros in the image configuration could produce on the conditional expectation on the spin at the origin. Using a transfer matrix approach, we prove that this effect is asymptotically absent (essential continuity), except for a few peculiar configurations\footnote{{See (\ref{alter}), these are the configurations where the unique path of zeros is neighbored by {alternating} $+/-$'s.}}. We evaluate also finite-volume effects in order to get afterwards a sufficient condition on essential continuity in next section.}

\begin{lemma}[Continuity of {conditional} magnetization]\label{thm.contmagn}
Let $\eta' \in \Omega'$ such that $N(\eta')=1$, except alternating configurations of the type of $ \eta'^1_{\rm alt},\eta'^2_{\rm alt}$ as defined in (\ref{alter}). Then the conditional magnetizations (\ref{condmagn1}) {are} essentially continuous at $\eta'$.
Moreover, there exists a strictly positive constant $C>0$ such that
\be \label{EstR}
\forall R \geq 0, \forall \omega'_1, \omega'_2 \in \mathcal{N}^{R}(\eta'), \qquad \mid \langle \sigma'_r \rangle^{\omega'_1,R}-\langle \sigma'_r \rangle^{\omega'_2,R} \mid \leq C \cdot (e^{-\beta})^R.
\ee
\end{lemma}

\begin{proof}

Let us consider,  for $\eta' \in \Omega'$ such that $N(\eta')=1$,
the set of sites $\pi(\eta')$ to be the (unique) infinite cluster from the origin in configuration $\eta'$. We restrict ourselves to configurations where the spin in the
cells in the neighborhood of the path is never zero and denote the path by $\pi$. If this is not the case, one can start the tree at the first such cell, and evaluate the magnetization at this site ``$r$''. As the labeling of the trees should not affect our results, we can also assume without loss of generality that $\pi$ is the set of sites having
binary representation ``$1$'', except the first one.

Since the length of all the other paths of zeros is {finite},  we can always consider an integer $R_1=R_1(\eta')$ which is the maximal length of these other paths  (similarly to the
definition of $R_0$ adopted in Section~\ref{sss.cn0}). {Consider} now a region which is sufficiently far away from the origin $r$, with $R \geq R_1$, {\em i.e.}\ where no such terminating path of zeros {penetrates}. Define also the projection
$Y'$ of $\eta'$ onto the only infinite path $\pi$ by $Y'_{R}=\eta'_{W_{R} \cap \pi}$,
{usefully extended} on the other neighbors of the
origin by
$
Y'_{-1}=\eta'_{0}.
$
 Write as usual $Y=T^{-1}(Y')$ and define the shortcut $X_n=Y_{R-n+1}$, for all integers $n \leq R$. We study the asymptotic
{behavior} of the $\nu$-magnetization when we have around the origin a
b.c.\ in a neighb{o}rhood of $\eta'$ :
\begin{displaymath}
\langle \sigma'_r \rangle^{\eta',R}=
\nu[\sigma'_r =+\mid A'_R]-\nu[\sigma'_r =-\mid A'_R],
\end{displaymath}
where
$$
A'_R=\{\sigma' \in \Omega', \sigma'_{V_{R} \setminus \{r\}}=\eta'_{V_{R} \setminus \{r\}}\}.
$$  
 Up to now, $\eta'_{0} \neq 0$ and we shall assume without loss of generality
that $\eta'_{0}=+$. This implies $\sigma_{0}=+$ and, with the usual notation  $A_R=T^{-1}(A'_R)$, we get
\begin{eqnarray*}
\langle \sigma'_r \rangle^{\eta',R}
&=& \mu[\sigma_r =\sigma_{0}=\sigma_{1}=+
\mid A_R]\\
&=&  \mu[\sigma_r =\sigma_{1}=+
\mid A_R \cap \{\sigma_{0}=+ \}]\mu[\sigma_{0}=+
\mid A_R]\\
&=&  \mu[\sigma_r =\sigma_{1}=+
\mid A_R \cap \{\sigma_{0}=+ \}].
\end{eqnarray*}
When $\eta'_{0}=-$, we {analogously} obtain $\langle \sigma'_r \rangle^{\eta',R}=\mu[\sigma_r=\sigma_{1}=-
\mid A_R \cap \{\sigma_0 =-\}]$. Thus we only have to study the behavior
of
\begin{displaymath}
\mu[\sigma_r =\sigma_{1}=+
\mid A_R \cap \{\sigma_{0}=+ \}] \; {\rm 
 and} \; 
 \mu[\sigma_r =\sigma_{1}=-
\mid A_R \cap \{\sigma_{0}=- \}],
\end{displaymath} 
knowing the configuration $\eta'$ along the
path is zero {under} the previous assumptions on $\eta'$: we have to study the law of $X$ with this environment $\eta'$.

\smallskip
  
Now, if $X_n$ is the spin at a site $n  \in \pi$, let $h_n$ be the {spin value} at
the neighb{o}r of $n$ which is not on the path 
%(see figure~\ref{fig.Xemfh}) 
and define $h=(h_n)_{n \in \mathbb{N}}$. All the possible values of $h$, which can be seen as an external field for the process $X$, are determined by the
configuration $\eta'$ but are independent of $R$ because of the
uniqueness of the path of zeros. Let us denote by $\mathbb{P}$ the probability measure $\mu[ .  \mid A_R]$. We shall study first the law of $X$ under $\mathbb{P}$.

\smallskip

Forget for a while the constraint due to the path of zeros and study
  the law of $X$ without it. With the
  fixed spin at the origin and a fixed external configuration (field)
  $h=h(\eta')=(h_n)_{n \in \mathbb{N}}$, the law of
  $X$ is exactly the law of a one dimensional Ising model at inverse temperature
  $\beta$ with coupling $J>0$ and external magnetic field $g$ defined
  by
\be \label{gR}
 g_{R+1}=0 \; \textrm{and} \; \forall n \in \mathbb{N}, \; g_n= \beta J h_n.
\ee

We shall now incorporate the coupling in the temperature, assume
$J=1$ and denote $\gamma$ the inhomogeneous specification of this
Ising chain on $\mathbb{Z}$ (see {\em e.g.} \cite{HOG} for details). 

\smallskip

From now on we can proceed \emph{via} a classical transfer matrix approach: define $\forall n \in \mathbb{N}$ a $2 \times 2$ matrix
$Q'_n$ by
\be \label{Qprime}
\forall n \in \mathbb{N}, \; \forall x,y \in \{-,+\}, \;
Q'_n(x,y)=\exp\big(\beta xy +\frac{\beta}{2}(h_{n-1}x+h_ny)\big);
\ee
and rewrite the specification in terms of this transfer
matrix: Define $V=\{i+1, \dots , k-1 \} \subset \mathbb{N}$ and let $\sigma, \omega \in
\{-,+\}^{\mathbb{N}}$. Then

\begin{displaymath}
\gamma_V(\sigma_V\mid
\omega)=\frac{Q'_{i+1}(\omega_i,\sigma_{i+1})(\prod_{j=i+2}^{k-1}Q'_j(\sigma_{j-1},\sigma_{j}))
  Q'_{k}(\sigma_{k-1},\omega_k)}{(\prod_{j=i+1}^{k}Q'_j)(\omega_i,\omega_k)}.
\end{displaymath}

%{\bw Transition, commenter et introduire les techniques}

We first derive  the
homogeneous case $h_n=+ \; \forall n \in \mathbb{N}$,
assuming  $\eta'_{0}=+$\footnote{The cases $\eta'_{0}=-$ or $h_n=- ,  \;\forall n \in
  \mathbb{N}$, are treated similarly.}, and afterwards for a general, {possibly inhomogeneous} environment $h$ (leading to
inhomogeneous Markov chains), without the constraints of being on the path in a first step, and imposing it in a second step.

\smallskip

{\bf Homogeneous external field case}

\smallskip

Let us assume $\eta'_{0}=+$, $h_n(\eta')=+ ,\; \forall n \in
\mathbb{N}$ and write $\eta'=\eta'^+$. We
have

\begin{displaymath}
\forall n \in \mathbb{N}, \; Q'_n=Q'=
\left(\begin{array}{cc} 1&
      e^{-\beta}\\ e^{-\beta} &  e^{2\beta}
\end{array}
\right).
\end{displaymath}

It is well known that \emph{there is a one to one correspondence
between the set of all positive homogeneous Markov specifications and
the set of all stochastic matrices on $E$ with no vanishing entries} \cite{HOG}.  Hence, 
under the environment $h$ (without the constraint on
the path of zeros), the law of $X$ is that of an homogeneous Markov
chain with a transition matrix 
%$P$
% and a priori law $\alpha_P$\footnote{$\alpha_P$ is the only probability %(row) vector such
%that $\alpha_P=\alpha_P P$.} such that
\begin{displaymath}
P=\left(\begin{array}{cc} a &
      1-a \\ 1-b &  b
\end{array}\right),
\end{displaymath}
{where} $a=\frac{e^{-\beta}}{\cosh{\beta}+\sqrt{e^{-\beta} +
    (\sinh{\beta})^2}} \; \in \; ]0,1[$ and $b=e^{2\beta}$.

\medskip

Let us now introduce the constraint of being on the path of zeros. The
law of $X$ is still Markovian, but  now some transitions are forbidden. In
this case with $h_n=+, \; \forall n \in \mathbb{N}$, the transition
matrix becomes
\begin{displaymath}
P_+=\left(\begin{array}{cc}
a & 1 - a \\
1  & 0
\end{array}
\right)
\end{displaymath}
because, with this constraint, necessarily
\begin{displaymath}
\{h_{n+1}=+,X_n=+\} \Longrightarrow \{X_{n+1}=-\}
\end{displaymath}
otherwise the infinite path of zero{e}s would end at cell $C_{n}$.

Now, recall that in order to study the continuity of the magnetization under the
environment {induced by $\eta'^+$}, we have to study the asymptotic behavior, on the
neighb{o}rhoods of $\eta'^+$ and when $R$
goes to infinity, of $\langle \sigma'_r \rangle^{\eta'^+,R}$. {Under} the
assumption of $\eta'_{0}=+$, we get
\be
 \label{eq.sprreR}
\begin{split}
\langle \sigma'_r \rangle^{\eta'^{+},R}=&\mu\big[\{\sigma_{1}=\sigma_r=+\}  \mid  A_R \cap
\{\sigma_0=+\}\big]\\
=&\gamma_{R}\big[X_R=X_{R+1}=+ \mid X_0=x_0\big]\\
=&\mu_P\big[X_R=X_{R+1}=+ \mid X_0=x_0\big]\\
=&\mu_P\big[X_{R+1}=+ \mid X_R=+\big]\mu_P\big[X_R=+\mid X_0=x_0\big]\\
=&P(+,+)(P_{+}^{R+1})_{x_0,+},\\
\end{split}
\ee
in accordance with the expression of the Ising random field as a Markov
chain\footnote{Here, $(.)_{x_0,+}$ denotes the column $+$ and line $x_0$ of the
matrix, with $x_0=+$ or $-$.}. On the first step, we keep the matrix $P$ because
there is no constraint. Thus, studying continuity of the magnetization is
reduced to the study of the dependence on $x_0$ of the 
expression in {(\ref{eq.sprreR})}. By the usual diagonalization {procedure}, we get for all  $n \in \mathbb{N}$:
\begin{displaymath}
P_{+}^{n}=\left(\begin{array}{ccc}
\frac{1}{2-a} + \frac{(a-1)^{n+1}}{2-a} & \frac{1 -
  a}{2-a}-\frac{(a-1)^{n+1}}{2-a} \\
\\
\frac{1}{2-a} + \frac{(a-1)^{n}}{2-a}  & \frac{1 - a}{2-a}-\frac{(a-1)^{n}}{2-a}
\end{array}
\right),
\end{displaymath}
and hence $M\coloneqq \lim_{n \to \infty} P_{+}^{n}$ is {simply}
\be \label{eq.Mmat}
M=\left(\begin{array}{ccc}
\frac{1}{2-a} & \frac{1 -
  a}{2-a} \\
\\
\frac{1}{2-a}  & \frac{1 - a}{2-a}
\end{array}
\right).
\ee
Notice that the
elements of each column in $M$ are equal, implying that $\lim_{R \to
  \infty}(P_{+}^{R+1})_{x_0,+}$ is independent on $x_0$: \emph{the magnetization is a continuous function of the
boundary condition}. It is given by
\begin{displaymath}
\langle \sigma'_r \rangle^{\eta'^{+}}=e^{2\beta}a\frac{1 - a}{2-a}
\end{displaymath}
for a configuration $\eta'^+$ with one infinite path of zeros and $h_n=+$ everywhere. For a configuration $\eta'^{-}$ with only one path of
zeros and $h_n=-$ everywhere, the same
continuity result holds, and we simply have
\begin{displaymath}
\langle \sigma'_r \rangle^{\eta'^{-}}=-e^{2\beta}a\frac{1 - a}{2-a} = -\langle \sigma'_r \rangle^{\eta'^{+}},
\end{displaymath}
where we recall that $a={a(\beta)}=\frac{e^{-\beta}}{\cosh{\beta}+\sqrt{e^{-\beta} +
    (\sinh{\beta})^2}} \; \in \; ]0,1[$. 
    
   % We have then described another good configuration.
    
    % Let us investigate inhomogeneous cases.
    
\smallskip

{\bf Inhomogeneous external field case}

\smallskip

Here, we deal with an inhomogeneous Markov Chain
without knowing exactly the transition matrix (we only know a
transfer matrix and the notion of boundary laws, see \cite{HOG}, {Def}.\ 12.10). With the path of zeros constraint, the law of $X$
is still Gibbsian but the potential now becomes

\begin{displaymath}
{\Psi}_A(\omega)=\left\{
\begin{array}{lllllllll}\; -\beta \omega_{n-1}\omega_{n} -\frac{\beta}{2}(h_n\omega_n+h_{n-1}\omega_{n-1}) \\
 \; \; \;  \; \; \;  \; \; \; \textrm{iff}\;
  A=\{n-1,n\} \;\textrm{and} \; h_n\omega_{n-1}=h_n\omega_n=+\\
\\
\; + \infty \; \; \textrm{iff} \;A=\{n-1,n\} \;\textrm{and}\; \omega_n=\omega_{n-1}=h_n\\
\\
\; 0 \; \;  \; \; \; \textrm{otherwise},
\end{array} \right.
\end{displaymath}

because some transitions are forbidden under the constraint. The potential can take infinite values and it provides a kind of {\em hard core exclusion potential}, but the formalism is the same as in the {finite potential case} (provided the objects are defined \cite{HOG}). We {then}
have to deal with various products of the matrices $Q_{h_nh_{n-1}}$
depending on $h$, {so that it is useful to introduce}:
\begin{displaymath}
Q_{++}=\left(\begin{array}{ccc}
1& e^{-\beta}\\
\\
e^{-\beta}  & 0
\end{array}
\right)
\; , \;
Q_{--}=\left(\begin{array}{ccc}
0 & e^{-\beta}\\
\\
 e^{-\beta} & 1
\end{array}
\right)
\; , \;
\end{displaymath}
\begin{displaymath}
 Q_{-+}=\left(\begin{array}{ccc}
0& e^{-2\beta}\\
\\
1 & e^{\beta}
\end{array}
\right)
\; \textrm{and}  \; \;
Q_{+-}=\left(\begin{array}{ccc}
e^{\beta}&1\\
\\
e^{-2\beta} & 0
\end{array}
\right).
\end{displaymath}
Keeping the same notations as in the homogeneous case, we first deal with the case $\eta'_0=+$. Let $x_0 \in \{-1,+1\}$ and define, for $R \geq 0$,
\begin{displaymath}
\langle \sigma'_r \rangle^{\eta',R}_{x_0}=\nu\big[\sigma'_r \mid \sigma'_{\{r\}^{c}}=\omega'_{\{r\}^{c}}, \omega' \in \mathcal{N}^{R}(\eta'), X_0= x_0 \big].
\end{displaymath}

Proceeding as usual, we get, in this case where $\eta'_0=+$,
\begin{eqnarray*}
\langle \sigma'_r \rangle^{\eta',R}_{x_0}&=&\nu\big[\sigma'_r=+ \mid \sigma'_{\{r\}^{c}}=\omega'_{\{r\}^{c}}, \omega' \in \mathcal{N}^{R}(\eta'), X_0= x_0 \big]\\
&=&\mu\big[X_R=X_{R+1}=+ \mid X_{R+2}=+,h,X_0=x_0\big],
\end{eqnarray*}
where the notation ``$h$'' means a conditioning with the event of being under the environment $h=(h_n)$ besides the path of zeros. Using again the expression {(\ref{Qprime})} with transfer matrices,
%\begin{displaymath}
%\langle \sigma'_r \rangle^{\eta,R}_{x_0}
%\end{displaymath}
\begin{eqnarray*}
\langle \sigma'_r \rangle^{\eta',R}_{x_0}&=&\frac{Q'_{R+2}(+,+)Q'_{R+1}(+,+)[Q_R \dots Q_n \dots Q_1]_{x_0,+}}{Q'_{R+2}(+,+)Q'_{R+1}(+,+)([Q_R \dots Q_n \dots Q_1]_{x_0,-}+Q'_{R+2}[Q_R \dots Q_n \dots Q_1]_{x_0,+})}\\
&=&\frac{[Q_R \dots Q_n \dots Q_1]_{x_0,+}}{[Q_R \dots Q_n \dots Q_1]_{x_0,-}+[Q_R \dots Q_n \dots Q_1]_{x_0,+}}
\end{eqnarray*}
where, due to the constraints and the fields (\ref{gR}) given by $g_{R+1}=0, g_{R+2}=2, g_n=h_n \; \forall  n \leq R$,
\begin{displaymath}
 Q'_{R+2}=\left(\begin{array}{ccc}
1& 1\\
\\
e^{-2\beta} & e^{2\beta}
\end{array}
\right)
\; \textrm{,}  \; \;
Q'_{R+1}=\left(\begin{array}{ccc}
e^{\beta - \frac{\beta h_R}{2}}&e^{-\beta - \frac{\beta h_R}{2}}\\
\\
e^{-\beta + \frac{\beta h_R}{2}}& e^{\beta + \frac{\beta h_R}{2}}&
\end{array}
\right)
\end{displaymath}
and
\begin{displaymath}
\forall n \leq R, \quad Q_n=Q_{h_n h_{n-1}}.
\end{displaymath}
 Thus, in order to study the continuity of this magnetization as a function of the boundary conditions $h_0,h_1$ and $x_0$, we have to study the asymptotic behavior, {depending on} $h_0,h_1,x_0$, of {the matrix products}
\begin{displaymath}
P_R=Q_R \dots Q_1=\prod_{n=R}^{1}Q_n=\prod_{n=R}^{1}Q_{h_nh_{n-1}}
\end{displaymath}
taking {into} account the constraints\footnote{The constraints imply that some products such as $Q_{++}Q_{-+}$ are forbidden.}. 
Denote
\begin{displaymath}
P_n=\left(\begin{array}{ccc}
a_n & b_n\\
\\
c_n & d_n
\end{array}
\right)
\end{displaymath}
and {assume first that}, there exists $n_0=n_0(\eta')< + \infty$ {s.t.~at least} for $n \geq n_0 \in \mathbb{N}$,  no entry of $P_{n}$ is zero, namely $a_n \cdot b_n \cdot c_n \cdot d_n >0$.
{This condition (finiteness of $n_0$) holds for all configurations} $\eta'$ except special 
{\em alternating configurations}  (see \eqref{alter}).

For such $\eta'$,  take now $n \geq n_0=n_0(\eta')$ and denote $x_n=\frac{a_n}{b_n}$ and $y_n=\frac{c_n}{d_n}$. We want to study the asymptotic  behavior of the sequences $x=(x_n)_{n \in \mathbb{N}}$ and $y=(y_n)_{n \in \mathbb{N}}$. We have {the general pattern} $P_{n+1}=P_n \times A_n$, with
\begin{displaymath}
A_n \in \{Q_{++},Q_{+-},Q_{--},Q_{-+}\}
\end{displaymath}
so that we can consider the four cases separately.

\begin{description}
\item[Case {$A_n = Q_{++}$}]: 
We obtain
\begin{displaymath}
P_{n+1}=\left(\begin{array}{ccc}
a_n + b_n e^{-\beta}& a_n e^{-\beta}\\
\\
c_n + d_n e^{-\beta}& c_n e^{-\beta}
\end{array}
\right)
\end{displaymath}
and this yields the same evolution for $(x_n)_{n \in \mathbb{N}}$ and $(y_n)_{n \in \mathbb{N}}$:
\begin{displaymath}
x_{n+1}=f_1(x_n) \; \textrm{and} \; y_{n+1}=f_1(y_n),
\end{displaymath}
where $f_1$ is defined for all $x >0$ by
\begin{displaymath}
f_1(x)=1+ \frac{e^{-\beta}}{x} \geq 1.
\end{displaymath}
We also have
\begin{displaymath}
\forall x \geq 1, \; \mid f'_1(x) \mid = \frac{e^{-\beta}}{x^2} \leq e^{-\beta} <1.
\end{displaymath}
This application is then contracting as soon as $x \geq 1$, which will be true after one step because $f_1(x) \geq 1 \; \forall x >0$. We denote the {Lipschitz} constant $k_1=e^{-\beta}$, and we have
\begin{displaymath}
\forall n \geq n_0, \; \mid x_{n+1} - y_{n+1} \mid \leq k_1 \cdot  \mid x_n -y_n \mid .
\end{displaymath}
\item[Case $A_n=Q_{+-}$]:
We obtain
\begin{displaymath}
P_{n+1}=\left(\begin{array}{ccc}
a_n e^{\beta} + b_n e^{-2\beta} & a_n \\
\\
c_n e^{\beta}+ d_n e^{-2 \beta}& c_n
\end{array}
\right)
\end{displaymath}
and the same contractive result holds with the function $f_2$ : 
\begin{displaymath}
f_2(x)=e^{\beta}+ \frac{e^{-2\beta}}{x} > e^{\beta} >1, {\rm for \; all \; } x >0.
\end{displaymath}
Now
\begin{displaymath}
\forall x \geq 1, \; \mid f'_2(x) \mid = \frac{e^{-2\beta}}{x^2} \leq e^{-2\beta} <1
\end{displaymath}
and we will denote by $k_2=e^{-2\beta}$ this second Lipschitz constant.
\item[Case $A_n=Q_{--}$]:
We obtain
\begin{displaymath}
P_{n+1}=\left(\begin{array}{ccc}
b_n e^{-\beta}& a_n e^{-\beta} + b_n\\
\\
d_n e^{-\beta}& c_n e^{-\beta} +d_n
\end{array}
\right)
\end{displaymath}
and 
$
f_3(x)=\frac{1}{x+e^{\beta}} ,\; \forall x > 0,
$
which gives a  Lipschitz constant $k_3=k_2=e^{-2\beta}$.
\item[Case $A_n=Q_{-+}$]:
The contraction holds  with  $k_4=k_2=e^{-2\beta}$.
\end{description}

\smallskip

Thus, whatever the configuration ${\eta'}$ different\footnote{In the sense that it should not contain infinite alternate paths like in $\eta'^{1/2}_{\rm alt}$ in (\ref{alter}). By abuse of notation, we write $\eta' \neq \eta'^1_{\rm alt},\eta'^2_{\rm alt}$ for this property.} from the alternating ones we always have
\be \label{eq.ubAQ}
\forall n \geq n_0, \; \mid x_n - y_n \mid \leq e^{n_0 \beta}\left(\frac{1}{e^{\beta}}\right)^n\cdot \mid x_{n_0}-y_{n_0} \mid
\ee
{for some finite $n_{0}=n_{0}(\eta')$}.
Let us come back to the magnetization (\ref{condmagn1}). Whatever $h_0$ and $h_1$ are, we have
\begin{displaymath}
\langle \sigma'_r \rangle^{{\eta'},R}_{x_0}=\frac{[P_R]_{x_0,+}}{[P_R]_{x_0,}+[P_R]_{x_0,+}},
\end{displaymath}
which yields
$
\langle \sigma'_r \rangle^{\eta',R}_{+}=\frac{d_R}{c_R + d_R} \; \textrm{and} \; \langle \sigma'_r \rangle^{\eta',R}_{-}=\frac{a_R}{a_R + b_R}.
$
If $n \geq n_0$, we get
\begin{displaymath}
\langle \sigma'_r \rangle^{\eta',R}_{+}=\frac{1}{1 + y_R} \; \textrm{and} \; \langle \sigma'_r \rangle^{\eta',R}_{-}=\frac{1}{1 +x_R}
\end{displaymath}
and
\be \label{DeltaR}
\mid \langle \sigma'_r \rangle^{\eta',R}_{+} - \langle \sigma'_r \rangle^{\eta',R}_{-} \mid \leq \mid x_R - y_R \mid \leq e^{n_o \beta} \cdot \left(\frac{1}{e^{\beta}}\right)^R \mid x_{n_0}-y_{n_0} \mid
\ee

Now, for all {$\eta' \neq \eta'^1_{\rm alt},\eta'^2_{\rm alt}$, write $$C(\eta') \coloneqq e^{n_0(\eta')(\beta)} \mid x_{n_0(\eta')}-y_{n_0(\eta')} \mid$$

so that $$C \coloneqq \sup_{\eta'\neq {\eta'^1_{\rm alt},\eta'^2_{\rm alt}}} C(\eta') < \infty$$ is well defined. This proves that (\ref{DeltaR}) converges to zero as $R \to \infty$,   proving also (\ref{EstR}).

\smallskip

Proceeding similarly in the case $\eta'_0=-$, and conditioning on the events $\{X_{R+2}=+\}$ and $\{X_{R+2}=-\}$, we get continuity of the magnetization for all the configurations having only one path of zeros {so that $n_0(\eta')<\infty$}, {\em i.e.}~except the so-called exceptional {\em alternating configurations} described below.  This proves Lemma 4.3.

\smallskip 
{\bf Remark }(Peculiar alternating configurations)  Assume again $\eta'_0=+$ and define  {$\eta'^1_{\rm alt}$} and { $\eta'^2_{\rm alt}$} to be such that $N(\eta')=1$ and
\be \label{alter} 
\forall i=1,2,\; \forall R \in \mathbb{N}, \; \forall n=0 \dots R-1, \quad h^i_R=(-1)^i,\; h^i_n=-h^i_{n+1}
\ee
where $h^1$ (resp.\ $h^2$) is the environment under {$\eta'^1_{\rm alt}$} (resp.\ $\eta'^2_{\rm alt}$). For such configurations, elementary analysis of the products of matrices considered in this section show that the limits considered do not exist, because sequences alternate between two accumulation points. Nevertheless,  we prove in Lemma \ref{4.7} that such configurations form a $\nu$-negligible set.}

\end{proof}

 \subsection{Sufficient condition for essential continuity }

Recall that, for a configuration $\eta'$, $N_{R}(\eta')$ denotes  the number of paths of zeros up to level $R$ in $\eta'$ and $N=N(\eta')=\lim_{R\to \infty}N_{R}(\eta')$ exists by monotonicity. Use Lemma \ref{thm.contmagn} to show
\begin{lemma}\label{lem.suffcondec}
Let $\eta' \in \Omega'$ with $N(\eta')$ arbitrary, different from the null configuration or the alternate configurations defined above in (\ref{alter}). Then 
\be\label{suffcondec}
\forall R \geq 0, \sup_{\omega'_1,\omega'_2 \in \mathcal{N}^R(\eta')}\mid \langle \sigma'_r \rangle^{\omega'_1,R}-\langle \sigma'_r \rangle^{\omega'_2,R} \mid \leq C \cdot N_{R}(\eta') \cdot (e^{-\beta})^R.
\ee
\end{lemma}

\begin{proof}

The proof consists in associating the recursive tree structure and the Markov property of the measure $\mu$ to the estimate (\ref{EstR}) of Lemma \ref{thm.contmagn} for each path of zeros of length $R$. 

First of all, let $\Pi(\eta')=\left(\pi_{k}\right)_{k=1}^{N_{R}(\eta')}$ be the set of path of zeros in $\eta'$ according to some (inessential) ordering, and let $\overline{\eta}'^{k}$ be the configuration coinciding with $\eta'$ everywhere, except at the $N_{R}(\eta')-1$ path of zeros different from the $k$-th one, that is, $\overline{\eta}'^{k}=\eta'_{\pi_{k}\cup \Pi^{c}}+_{\pi_{k}^{c} \cup {\Pi}}$. Let $n_{k}(\eta')$ be the smallest integer such that no entry of the matrix $P_{n}$ along $\pi_{k}$ is zero; that is, accordance with our notations in Lemma \ref{thm.contmagn}, $n_{k}(\eta')=n_{0}(\overline{\eta}'^{k})$.  

If $N_{R}(\eta')=1$ the claim just coincides with Lemma \ref{thm.contmagn}. Let us now prove it for $N_R(\eta')=2$. First, observe that, for any $\omega_{1}',\omega_{2}' \in \mathcal{N}^{R}(\eta')$, by Markov property and splitting on the tree,
$$
\forall R > 0,  \left|\magnprime{r}{\omega'_{1},R}-\magnprime{r}{\omega'_{2},R}\right| \leq p(\beta) \left|\magnprime{r1}{\omega'_{1},R}\magnprime{r0}{\omega'_{1},R}-\magnprime{r1}{\omega'_{2},R}\magnprime{r0}{\omega'_{2},R}\right|
$$
where $p(\beta) \coloneqq \nu[\eta'_{0}=0] =\frac{2+e^{-2\beta}}{\left(e^{-\beta}+e^{\beta} \right)^{2}}$ (see Section 4.4).

 Second, remark that any four numbers $x,y,z,w \in [0,1]$ we have 
\be
\begin{split}
|xy-zw|&=|xy-yz+yz-zw|=|y(x-z)+z(y-w)|\underset{\rm tr. ineq.}{\leq} y|x-z| + z|y-w|\\
&\underset{y,z \in [0,1]}{\leq} |x-z| + |y-w| \; .
\end{split}
\ee
Putting
$$
\begin{cases}
x &= \magnprime{r0}{\omega'_{1},R}\\
y &= \magnprime{r1}{\omega'_{1},R}\\
z &= \magnprime{r0}{\omega'_{2},R}\\
w &= \magnprime{r1}{\omega'_{2},R}\\
\end{cases}
$$
gives
\be
\label{eq.fundineq}
\left|\magnprime{r1}{\omega'_{1},R}\magnprime{r0}{\omega'_{1},R}-\magnprime{r1}{\omega'_{2},R}\magnprime{r0}{\omega'_{2},R}\right| \leq \left|\magnprime{r0}{\omega'_{1},R}-\magnprime{r0}{\omega'_{2},R}\right|+\left|\magnprime{r1}{\omega'_{1},R}-\magnprime{r1}{\omega'_{2},R}\right|.
\ee

Now, apply Lemma \ref{thm.contmagn} to each of the quantities of the rhs, we get 
 $$
 \forall R \geq 0, \sup_{\omega'_1,\omega'_2} \mid \langle \sigma'_r \rangle^{\omega'_1,R}-\langle \sigma'_r \rangle^{\omega'_2,R} \mid \leq 2 \, C \, (e^{-\beta})^R,
 $$

which is the claim for $N_{R}(\eta')=2$.
The result for general (finite) $N_R(\eta')$ is obtained analogously by iteration.
\end{proof}

\subsection{Percolation model and  almost Gibbsianness}

In the previous sections we have seen that the number of infinite paths of zeros is a good discriminant for detecting failures of quasilocality. Let us consider now the {marginal  probability of being $0$ for the image configuration $\eta'$}.  Due to the overlap among neighboring {cells} involved in the majority rule transformation $T$, clearly the random variables $( \eta'_{j})_{j \in \rbt}$ are not independent. Nevertheless, we can obtain a few informations about {$\eta'_{j}$'s} due to the Markovianness\footnote{{Strictly speaking, the standard (one-sided) Markov chain property is valid only for the Gibbs measures which are indeed Markov chains, see \cite{HOG, hig, spi, zach}. Nevertheless, this is true for any extremal Gibbs measures, and once our result is proven for any extremal $\mu$, it is straightforward to extend to all Gibbs measures.}} of the original measure. For a site $j \in \rbt$, denote $j1$ and $j0$ its descendants using the binary representation introduced in Section 3. When the conditioning is possible, for $\eta' \in \Omega'$, from Bayes's formula we get
\be
\begin{split}
\label{eq.grt1}
\nu\big[\eta'_{j1}=0\big]=&  \nu\big[\eta'_{j1}=0 \mid \eta'_j=+\big]\nu\big[\eta'_{j}=+\big] + \nu\big[\eta'_{j1}=0 \mid \eta'_j=-\big]\nu\big[\eta'_{j}=-\big]+\\
+& \nu\big[\eta'_{j1}=0 \mid \eta'_j=0\big]\nu\big[\eta'_{j}=0\big]
\end{split}
\ee
(if a conditioning is not possible, an analogous formula holds without their contributions). Let us now compute separately the three different conditional probabilities appearing in (\ref{eq.grt1}). In order to do so, let us introduce the shortcut 
$$
\mu_{\pm \pm \mid \pm}\coloneqq  \mu\big[\sigma_{j0}=\pm,\sigma_{j1}=\pm \mid \sigma_{j}=\pm\big].
$$
First, for the event $\{\eta'_{j}=+ \}$, 
\be
\bsplit
 \nu\big[\eta'_{j1}=0 \mid \eta'_j=+\big]=\;&
 \; (\mu_{+ - \mid +} + \mu_{- + \mid +}) + \mu_{- - \mid +}=\frac{2}{\left(e^{-\beta}+e^{\beta} \right)^{2}} + \frac{e^{-2\beta}}{\left(e^{-\beta}+e^{\beta} \right)^{2}} \\
 =\;&  \frac{2+e^{-2\beta}}{\left(e^{-\beta}+e^{\beta} \right)^{2}}.
\end{split}
\ee

 Analogously for the event $\{\eta'_{j}=- \} $ we get
 \be
 \bsplit
 \nu\big[\eta'_{j1}=0 \mid \eta'_j=-\big]=\;& 2\mu_{+- \mid -}+\mu_{++ \mid -}
 =\; \frac{2+e^{-2\beta}}{\left(e^{-\beta}+e^{\beta} \right)^{2}}.\\
\end{split}
 \ee

Lastly, the contribution due to conditioning on $\{\eta'_{j}=0 \}$ becomes
 \begin{eqnarray*}
  \nu\big[\eta'_{j1}=0 \mid \eta'_j=0\big]
 &=& \frac{2+e^{-2\beta}}{\left(e^{-\beta}+e^{\beta} \right)^{2}}\left(\mu[\sigma_{j}=+]+\mu[\sigma_{j}=-] \right) 
 =\frac{2+e^{-2\beta}}{\left(e^{-\beta}+e^{\beta} \right)^{2}},
\end{eqnarray*}
where the factors $2$ come from $+/-$ symmetry in a progeny at fixed father. Thus, despite the spins being dependent, the three considered events are indeed uncorrelated and we just have
\begin{displaymath}\label{eq.p2b}
 \nu[\eta'_{j}=0] \coloneqq p(\beta) =\frac{2+e^{-2\beta}}{\left(e^{-\beta}+e^{\beta} \right)^{2}} \quad \forall j \in \rbt .
\end{displaymath}
{Therefore, the problem becomes one of bond percolation on a Cayley tree \emph{with a $\beta$-dependent probability of open bond} (see {\em e.g.}~\cite{Grim}, Chapter 10)}. {In the case of our binary tree  $\mathcal{T}^{2}_0$, $N(\eta')$ being the number of infinite clusters of zeros in the configuration $\eta'$, we thus have}
\begin{displaymath}
\nu[N(\eta') > 0]=\left\{
\begin{array}{lllll}\; 0 \; & \textrm{iff}& p(\beta) \leq  p_{\rm c} = \frac{1}{2} \\
\\
\; >0 \;& \textrm{iff}& p(\beta)> p_{\rm c} = \frac{1}{2}  . \\
\end{array} \right. 
\end{displaymath}
In particular, we get almost Gibbsianness at high temperatures $\beta \leq  \beta_{1}$ with $\beta_1^2=\ln{\left(1+\sqrt{2}\right)}$.

\smallskip

Let us now study almost {Gibbsianness} in the low temperature regime $\beta > \beta_{1}$. By means of the sufficient condition of Lemma \ref{lem.suffcondec}, our aim is {first} to investigate the $\nu$-measure of the set $\Omega_{g}$ of configurations for which $N_{R}(\eta')$ {(the number of paths of zeros between the root and generation $R$ in configuration $\eta'$)} grows slower than $e^{\beta R}$:
\be
\Omega_{g}\coloneqq \left\{ \eta' \in \Omega' \; : \; \lim_{R \to \infty} \frac{N_{R}}{e^{\beta R}}=0  \right\}.
\ee

\begin{lemma}\label{Omegag}
Let $\mu$ be any Gibbs measure for the Ising model on $\mathcal{T}_0^2$ and $\nu=T \mu$. Then
$$
\nu(\Omega_g) = 1.
$$
\end{lemma}

%{\bf Proof of Lemma \ref{Omegag}:}}
\begin{proof} 
To proceed, we prove $\lim_{R \to \infty}\nu\big[\{ N_R(\eta') > e^{\beta R}\} \big]]=0$ so that typically $\lim_R \frac{N_{R}}{e^{\beta R}}=0$.

 To this aim, let us consider the sequence $(N_{R})_{R}$ of random variables on $(\Omega',\mathcal{F}')$. As $R$ grows, there is a probability $p(\beta)^{2}$ of increasing $N_{R}$ by $1$ ({\em i.e.}\ opening two new bonds), {a probability $(1-p(\beta))^{2}$ of closing two bonds, etc}. Therefore, for $\nu$-a.e. $\eta'$,
\be
\label{eq.recurexpNr}
\bsplit
 \mathbb{E}_{\nu}\left[N_{R} \mid \mathcal{F}_{V_{R-1}}\right](\eta')  =\;& p^{2}\left(N_{R-1} +1\right)+2p(1-p) N_{R-1} + (1-p)^{2}\left(N_{R-1} -1\right)\\
 =\;& N_{R-1}(\eta')+2p-1.
\end{split}
 \ee
By induction we simply get $\mathbb{E}_{\nu}\left[N_{R} \right] = p +\left(R-1\right)\left(2p-1\right)$
and thus 
\be \label{Expect}
\lim_{R\to \infty} \frac{\mathbb{E}_{\nu}\left[N_{R} \right] }{e^{\beta R}}=0.
\ee
\end{proof}
To go beyond expectations, let us now bound the probability of a deviation larger than $e^{\beta R}$.
In the context of percolation on trees, at this point one usually exploits the special tree topology to obtain recursive relations. An example is the determination of the critical temperature of the ferromagnetic Ising model with constant interaction strength (see \cite{Lyons89}, Theorem 2.1).  Instead, here we obtain a (tight) bound on the probability of an exponential deviation following a slightly different route based on a combination of an exponential Chebyshev's inequality and a uniform bound on the cumulant generating function of $N_{R}$:

\begin{lemma}[Exponential deviation bound for $N_{R}$]\label{thm.expdb}
For any $\theta \geq 0$
$$
\nu[N_{R}(\eta')>e^{\beta R}] \leq e^{\theta \left(R-e^{R\beta}\right)}
$$
so that only sub-exponential deviations are allowed at any finite temperature $\frac{1}{\beta}$ as $R \to \infty$.
\end{lemma}
\begin{proof}
For $R > 0$, exponential Chebyshev's inequality implies, $\forall \theta \in \R$,
\be
\label{eq.tch1}
\nu[N_{R}(\eta')>e^{\beta R}] \leq \mathbb{E}_{\nu}\left[e^{\theta N_{R}(\eta')} \right]e^{-\theta e^{R \beta}}.
\ee
In analogy with {(\ref{eq.recurexpNr})} we just have
\be
\bsplit
\mathbb{E}_{\nu}\left[e^{\theta N_{R}(\eta')} \mid \mathcal{F}_{R-1} \right] &= p^{2}e^{\theta (N_{R-1}+1)}+2p(1-p)e^{\theta N_{R-1}} + (1-p)^2 e^{\theta (N_{R-1}-1)} \\
&= \left(p^{2}e^{\theta}+2p(1-p)+(1-p)^{2}e^{-\theta} \right)e^{\theta N_{R-1}},\\
\end{split}
\ee
and hence by recurrence the full {\em moment generating function} (MGF) is given by
$$
 \mathbb{E}_{\nu}\left[e^{\theta N_{R}(\eta')}\right] \coloneqq M_{R}(\theta,p)= \left(p^{2}e^{\theta}+2p(1-p)+(1-p)^{2}e^{-\theta} \right)^{R}.
$$
A straightforward calculation shows that, for $\theta \geq 0$, the {\em Cumulant} $K_{R}(\theta ,p) \coloneqq \ln{M_{R}(\theta,p)}$
 satisfies, uniformly in $\theta$,
$$
K_{R}(\theta,0)=-R\theta \leq K_{R}(\theta,p) \leq R\theta=K_{R}(\theta,1),
$$
 from which the assertion follows directly (recall that $p=p(\beta)\in [0,1]$). 
\end{proof}

From Lemma \ref{thm.expdb} one gets that $\nu[N_R(\eta') > e^{\beta R}]$ converges to $0$ more than exponentially in the depth $R$, namely as $\sim e^{-e^{R}}$, which together with (\ref{Expect}) implies  Lemma \ref{Omegag} by {elementary probabilistic arguments. Indeed, for $l=e^{\beta}$ and $\epsilon=\epsilon_n=\frac{1}{2^n}>0$, we have
\begin{eqnarray*}
\nu(\Omega_g^c) &=&  \nu \left[ \cup_n \left\lbrace \liminf_R \frac{ N_R(\eta')}{l^{R}} > \epsilon_n \right\rbrace \right]
 \leq   \liminf_{R} \sup_\epsilon \epsilon \, e^{-l^{R}}\underset{R\to\infty}{\to} 0\;.
\end{eqnarray*}

Now, we consider the integers $n_k$ defined in Section 4.3 and introduce the sets 
 $$
\Omega_k \coloneqq \left\{ \eta' : n_k(\eta') < \infty \right\}.
 $$
 Consider the subset of $\Omega_g$ 
$$
\Omega_{f} \coloneqq \left \lbrace  \, \lim_{R \to \infty} \frac{N_{R}(\eta')}{e^{\beta R}}=0 \;{\rm and} \; n_{k}(\eta') < \infty, k=1,\ldots , N_{R}(\eta') \right \rbrace  = \Omega_g  \cap  \big( \cap_k \Omega_k \big).
$$

 We prove now that it is also $\nu$-typical.

\begin{lemma}\label{4.7}
$$\nu [ \Omega_{f} ]=1.$$
\end{lemma}
\begin{proof}
Consider the event that, up to level $R-1$, the configuration $\eta'$ contains a path of zeros along sites $\pi$ and alternating field $h_n=(-1)^n$ up to level $R-1$ at primed neighbors of $\pi$, that is, in the notations of Lemma \ref{thm.contmagn}, the event
$$
\mathcal{A}_{R-1}(\eta') = \left\lbrace h_{n} = (-1)^{n}, \; \forall n=1,\ldots,R-1 {\rm  \; at\; neighboring\; sites\; of\; a\; path\; of\; zeros}\; \pi \right\rbrace.
$$
Then
$$
\mathbb{P}\left[h_{R}\cdot h_{R-1}=-1 | \mathcal{A}_{R-1}(\eta') \right]_{\beta} \leq \mathbb{P}\left[h_{R}\cdot h_{R-1}=-1 | \mathcal{A}_{R-1}(\eta') \right]_{\beta=0}  = \left(\frac{1}{3}\right)^{2}+\left(\frac{1}{3}\right)^{2}<  p_{c},
$$

where we have used the infinite temperature limit gives an upper bound by attractivity, and $p_c=p_{c}(2)=\frac{1}{2}$ is critical threshold for Bernoulli percolation on a Cayley tree order $2$.
Thus, if $h_{R-1}=+1$ with $\eta'_{\pi_R}=0$ then there are $3$ choices of non-primed spins in $\eta'_{\pi_{R-1}}$, and $3$ choices in $\eta'_{\pi_R}=0$  (Fig.~\ref{fig.zigzag}). We can thus dominate the probability of this event of $h_R$'s being alternating (also called ``zebra percolation'', see~\cite{GRR13}) by a subcritical percolation event, with parameter $q<1/2$, and this probability goes to $0$ as $R \to \infty$. This argument is readily extended to each path of zeros since, after conditioning with our events, the subtrees become independent. Thus not only $\nu[n_0<\infty]=1$ but also $\nu[\Omega_k]=\nu[n_k<\infty]=1$, $\forall k \geq 1$.

 Eventually, by independence combined with $\nu[\Omega_g]=1$, we finally get
$$
\nu[\Omega_f]=\nu\big[\Omega_g \cap \big( {\displaystyle\cap_k} \{n_k < \infty\} \big)\big] =  \nu[\Omega_g] \cdot \prod_k \nu[\Omega_k] = 1\cdot \prod_k 1=1.
$$

\begin{figure}[!hbtp]
\centering
\includegraphics[width=.75\textwidth]{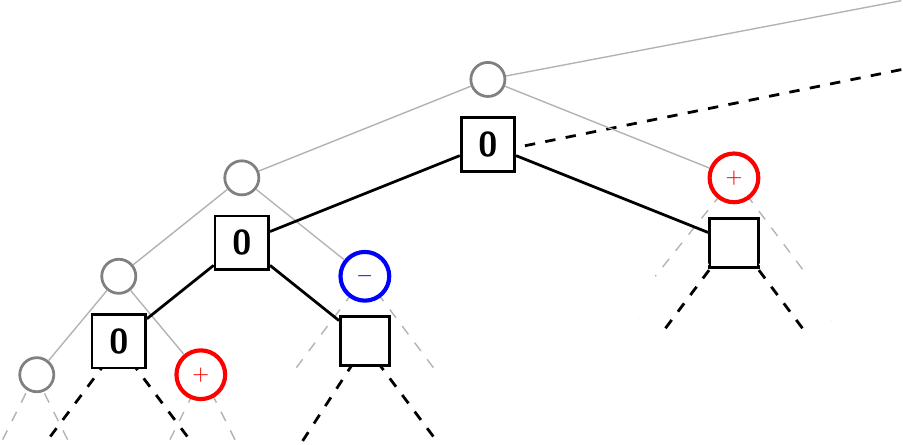} 
\caption{Slice of a configuration $\eta'$ at generations $R-1,R, R+1$. %Here $h_{R-1}=+$ (red) and $h_{R}=-$ (blue).
}
\label{fig.zigzag}
\end{figure}
\end{proof}

Now, write $\Omega_c$ for the set of good configurations for the conditional magnetization,
$$
\Omega_c = \{ \eta' \in \Omega' : \langle \sigma'_r \rangle^\cdot \; {\rm is \; essentially \; continuous} \}.
$$
\begin{lemma} \label{Omegac}
$$
\Omega_f \subset \Omega_c.
$$
\end{lemma}

\begin{proof} 

This is direct from the definition of our sets and previous lemmata. Consider $\eta' \in \Omega_f$, so that in particular $C(\eta') < \infty$, as well as $C=\sup_\eta' C(\eta')<\infty$. One has
$$
\frac{N_R(\eta')}{e^{\beta R}} \underset{R \to \infty}{\longrightarrow} 0
$$
implying essential continuity by (\ref{suffcondec}), so that $\eta' \in \Omega_c$. 

\end{proof}

Our percolation techniques yields thus the $\nu$-a.s.~essential continuity of conditional magnetizations and Lemma 5.

%\end{proof}
The extension of  this almost-sure continuity to the one of conditional expectations of any local function is straightforward by decomposition into simple functions, in the framework of Georgii~\cite{HOG}, Chapter 2, and Theorem~\ref{thm3} follows.

%\end{proof}
\section{Possible extensions and research perspectives}
\label{sec.rp}

{It would be interesting to understand if our almost Gibbsianness result holds for Cayley trees of general order $k\geq 2$. 
%\footnote{We remark that $p(\beta)=P_{2}(\tanh{\beta})$, where $P_{2}(y)=\frac{1}{4}(1-y)(y+3)$, and we believe this  should be a polynomial in $\tanh{\beta}$ for general Cayley trees of order $k \geq 2$.}
Concerning the crucial step in Theorem \ref{thm.expdb}, here we only observe that the MGF of the number of zeros reaching level $R$ in primed configuration $\eta'$, $N^{(k)}_{R}(\eta')$, is 
$$
M^{(k)}_{R}(\theta,p^{(k)})=\left[e^{-\theta} \left(p^{(k)} \left(e^\theta-1\right)+1\right)^k\right]^R,
$$
where now $p^{(k)}\equiv p^{(k)}(\beta)$ is the probability that a primed $0$ percolates at inverse temperature $\beta$ for the Cayley tree $\mathcal{T}^{k}$ ({\em i.e.}\ generalizing (\ref{eq.p2b})), and a $k$ {\em vs} $\beta$ tradeoff becomes possible.  Recently, ferromagnetic Ising models on Cayley trees subjects to inhomogeneous external fields have attracted some interest that could be useful for our purposes (see {\em e.g.}\ \cite{BEvE} for a recent work concerning spatially dependent external fields that are ``small perturbations'' of the critical external field value).

{Moreover,} the model considered in this paper can be naturally generalized
to non-rooted Cayley trees and to non-uniform $c_{j}$ possibly different from two, or with other sizes of {cells}. The majority rule could also be generalized as a stochastic transformation. {For example,} let $\epsilon \in [0,1]$ and $\xi$ be a Bernoulli random
variable with parameter $\epsilon$ and values $0$ or $1$. Define the
deterministic map $t_{\epsilon} \colon \Omega \longrightarrow \Omega';  \omega \longmapsto \omega'$
where $\omega'$ is defined by
\begin{displaymath}
\omega'_j=\left\{
\begin{array}{lllll}\; +1 \; & \textrm{iff}& \frac{1}{c}\sum_{i \in C_j}
\omega_i=+1 \; \textrm{and} \; \xi = 0 \\
%\\
\; -1 \;& \textrm{iff}& \frac{1}{c}\sum_{i \in C_j}
\omega_i=-1 \; \textrm{and} \; \xi = 0 \\
%\\
\; 0 \; & \textrm{otherwise.}
\end{array} \right.
\end{displaymath}
Its action is described by a probabilistic kernel $T_{\epsilon}$ defined by:
\begin{displaymath}
\forall A' \in \mathcal{F}', \forall \omega \in \Omega, \quad T_{\epsilon}(\omega,A')=(1-\xi)\delta_{t_{\epsilon}(\omega)}(A')+ \xi\delta_{0}(A') .
\end{displaymath}
 It could be interesting to study the difference between the deterministic transformation and the stochastic one, as this could play a role on the degree of non-Gibbsianness of the image measure, similar to the van den Berg example on the integers  (see \cite{LOR1}).
 
\medskip

{\bf Acknowledgements:} {We thank C. Maes for its early interest and A. van Enter for stimulating discussions and a careful reading of the manuscript. We also thank an anonymous referee for its constructive criticism in the previous version of this work. The research of the authors has been partially supported by CNRS, within the Franco-Dutch IRP B\'ezout-Eurandom.}

\end{document}